\documentclass[draftclsnofoot,onecolumn,12pt,romanappendices]{IEEEtran}

\usepackage{graphicx}

\usepackage{amssymb}
\usepackage{amsmath}

\usepackage{url}
\usepackage[cmex9]{subeqnarray}
\usepackage{cases}

\usepackage[bookmarks=false,dvipdfm,colorlinks,,linkcolor=blue]{hyperref}

\usepackage[thmmarks]{ntheorem}
{
\newtheorem{thm}{Theorem} 
\newtheorem{defn}{Definition} 
\newtheorem{lem}{Lemma} 
}
\newenvironment{proof}{\noindent{\textbf{Proof}}}{\hfill $\square $ \vskip 4mm}

\usepackage{enumerate}

\usepackage{algorithm}
\usepackage{algorithmic}

\begin{document}




\title{An Online Decoding Schedule Generating Algorithm for Successive Cancellation Decoder of Polar Codes}

 \author{\IEEEauthorblockN{Dan~Le, ~Xiamu~Niu$^*$}\\
 \IEEEauthorblockA{The School of Computer Science and Technology, Harbin Institute of Technology, Harbin, 150001 China.}
 }

\maketitle

\begin{abstract}
Successive cancellation (SC) is the first and widely known decoder of polar codes, which has received a lot of attentions recently. However, its decoding schedule generating algorithms are still primitive, which are not only complex but also offline. This paper proposes a simple and online algorithm to generate the decoding schedule of SC decoder. Firstly, the dependencies among likelihood ratios (LR) are explored, which lead to the discovery of a sharing factor. Secondly, based on the online calculation of the sharing factor, the proposed algorithm is presented, which is neither based on the depth-first traversal of the scheduling tree nor based on the recursive construction. As shown by the comparisons among the proposed algorithm and existed algorithms, the proposed algorithm has advantages of the online feature and the far less memory taken by the decoding schedule.

\end{abstract}




\section{Introduction}
\label{sec-Introduction}
Since Shannon presented the noisy channel coding theorem \cite{shannon1948bell}, polar code, introduced by Arikan \cite{arikan2009channel}, is the first class of codes achieving channel capacity with explicit construction. With the successive cancellation (SC) decoder, the channel capacity is asymptotically achieved by codelength $N$. Hence, polar codes have attracted many attentions recently\cite{hussami2009performance, tal2011list, niu2012stack, Alamdar2011A, leroux2011hardware, Chuan2012Reduced, huang2012latency, Leroux2013A, zhang2013low,Sarkis2013Increasing, Sarkis2014Fast, YouZhe2014An, Le2015Multi, Yoo2015Partially}.

As the first and widely known decoder of polar codes, a lot of research efforts have been made on SC decoder\cite{Alamdar2011A, leroux2011hardware, Chuan2012Reduced, huang2012latency, Leroux2013A, zhang2013low,Sarkis2013Increasing, Sarkis2014Fast, YouZhe2014An, Le2015Multi}. Some references focus on the simplified successive cancellation (SSC)\cite{Alamdar2011A, Chuan2012Reduced, Sarkis2013Increasing, Sarkis2014Fast, Le2015Multi}, which simplifies the constituent code with rate zero in SC decoder. SSC can significantly reduce the decoding latency and implementation complexity of SC decoder, but its performance depends strongly on the underlying channel and it requires that all frozen bits must be zeros. However, in some scenarios, the frozen bits can not be zeros, for instance, the error reconciliation in the quantum key distribution\cite{martinez2013key, qiong2014study}. Since SC decoder asks no restrictions on the frozen bits and its performance does not rely on the underlying channel, lots of works\cite{leroux2011hardware, Leroux2013A, zhang2013low, YouZhe2014An} are still done on the SC decoder. \cite{leroux2011hardware} presented an efficient hardware implementation of SC decoder with $O\left(N\right)$ processing elements and memory elements. \cite{Leroux2013A} proposed a semi-parallel SC decoder for resource sharing and processor sharing at the cost of a small increase in decoding latency. \cite{zhang2013low} showed a look-ahead and overlapped architectures to decrease the decoding latency of SC decoder. \cite{YouZhe2014An} proposed an efficient partial sum network architecture to reduce the decoding latency and implementation complexity for semi-parallel SC decoder. Although so many works have been done on SC decoder, its decoding schedule generating algorithms are still primitive. As far as we know, there are just two existed decoding schedule generating algorithms. One is based on the depth-first traversal of the scheduling tree\cite{Alamdar2011A,huang2012latency,Sarkis2013Increasing,YouZhe2014An,Le2015Multi}. The other is based on the recursive construction\cite{Chuan2012Reduced}. However, the problems are that they not only are too complex, but also generate decoding schedule offline and store it in the ROM. To overcome the problems, based on the newly found factor $z_i$, this paper proposes a new algorithm to generate the decoding schedule of SC decoder. The presented algorithm is more simple, obtains the decoding schedule on the fly without introducing any extra delay, and decreases the memory storing the decoding schedule significantly. These advantages reduce the implementation complexity of SC decoder.

The remainder of this paper is organized as follows. In Section \ref{sec-Traditional SC Decoding and Some Notations} we briefly review the SC decoder and introduce some notations. Section \ref{sec-Derive theoretically from recursive formula} explores the dependencies among likelihood ratios (LR), which lead to the discovery of the sharing factor $z_i$. Based on $z_i$, Section \ref{sec-Proposed Algorithm to Obtain Decoding Schedule} presents the proposed decoding schedule generating algorithm. Section \ref{sec-Comparisons} shows the comparisons among the proposed algorithm and existed algorithms. Finally, some conclusions are drawn in Section \ref{sec-Conclusion}.

\section{\label{sec-Traditional SC Decoding and Some Notations} SC Decoder and Some Notations}
First of all, let us list some notations used in this paper,
\begin{itemize}
\item $N=2^n$ is the code length of polar code, and $n = \log_2{N}$
\item $u_1^N$ is a shorthand for a row vector $\left(u_1, \cdots, u_N \right)$, and $u_i^j$, $1 \le i,j \le N$, represents its subvector $\left(u_i, \cdots, u_j \right)$
\item $\left\{a,\cdots,b\right\}$ represents the set of the integers ranging from $a$ to $b$
\item $\&$ is bitwise logical AND operator.
\end{itemize}

Polar codes take advantage of the polarization effect to achieve the channel capacity $I\left(W\right)$, whose channel model is illustrated as Figure \ref{Fig-Polar-Code-Channel-Model}, where $u_1^N$ is the input vector, $W_N$ is a combined channel by $N$ independent copies of channel $W$, and $y_1^N$ is the  output vector with conditional probability ${W_N}\left( {y_1^N|u_1^N} \right)$. For the coordinate channels $W_N^{\left( i \right)}$ of $W_N$, the size of the set $\left\{ {W_N^{\left( i \right)}\left| {I\left( {W_N^{\left( i \right)}} \right) \approx 1,\;1 \le i \le N} \right.} \right\}$ approaches $N \cdot I\left(W\right)$, while the size of the set $\left\{ {W_N^{\left( i \right)}\left| {I\left( {W_N^{\left( i \right)}} \right) \approx 0,\;1 \le i \le N} \right.} \right\}$ approaches $N \cdot \left( 1 - I\left(W\right) \right)$. When sending data, only the good coordinate channels are employed, which are called information bits. The indices set of information bits are denoted as $\cal A$, whose size is denoted as $K$. The set of other indices is named as ${{\cal A}^c}$, on which the values are called frozen bits, denoted as ${u_{{{\cal A}^c}}}= \left( {{u_i}|i \in {{\cal A}^c}} \right)$. The frozen bits ${u_{{{\cal A}^c}}}$ are known by both sender and receiver. Hence polar codes are usually denoted as $\left( {N,K,{\cal A},{u_{{{\cal A}^c}}}} \right)$.

\begin{figure}[htbp]
\centering
\includegraphics[width=0.3\textwidth]{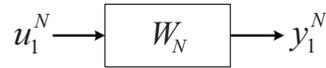}
\caption{\label{Fig-Polar-Code-Channel-Model} The illustration of the channel model of polar codes.}
\end{figure}
%

When decoding, SC decoder successively estimates the transmitted bits ${\widehat u_1^N}$ as follows,
\begin{eqnarray}
{\widehat u_i} = \left\{ {\begin{array}{*{20}{l}}
{{u_i},}&{{\rm{if}}\;i \in {{\cal A}^c}}\\
{0,}&{{\rm{if}}\;i \notin {{\cal A}^c}\;{\rm{and}}\;L_N^{\left( i \right)}\left( {y_1^N,\widehat u_1^{i - 1}} \right) \ge 1}\\
{1,}&{{\rm{if}}\;i \notin {{\cal A}^c}\;{\rm{and}}\;L_N^{\left( i \right)}\left( {y_1^N,\widehat u_1^{i - 1}} \right) < 1}
\end{array}} \right.
\label{eq-SC-determine}
\end{eqnarray}
, where
\begin{eqnarray}
L_N^{\left( i \right)}\left( {y_1^N,\widehat u_1^{i - 1}} \right) = \frac{{W_N^{\left( i \right)}\left( {\left. {y_1^N,\widehat u_1^{i - 1}} \right|0} \right)}}{{W_N^{\left( i \right)}\left( {\left. {y_1^N,\widehat u_1^{i - 1}} \right|1} \right)}}.
\label{eq-SC-L-N-i}
\end{eqnarray}
\eqref{eq-SC-L-N-i} can be straightforwardly calculated using the recursive formula
\begin{subequations}
\begin{small}
\begin{numcases}{L_N^{\left( i \right)}\left( {y_1^N,\hat u_1^{i - 1}} \right) = } \label{eq-SC-recursive-formula-LR-a}
{f\left( {L_{N/2}^{\left( {\left\lceil {i/2} \right\rceil } \right)}\left( {y_1^{N/2},\hat u_{1,o}^{i - 1} \oplus \hat u_{1,e}^{i - 1}} \right),L_{N/2}^{\left( {\left\lceil {i/2} \right\rceil } \right)}\left( {y_{N/2 + 1}^N,\hat u_{1,e}^{i - 1}} \right)} \right),} &${{\text{when}}\;i\;{\text{is}}\;{\text{odd}}}\;\;\;\;\;\;\;\;\;$ \\\label{eq-SC-recursive-formula-LR-b}
{g\left( {L_{N/2}^{\left( {\left\lceil {i/2} \right\rceil } \right)}\left( {y_1^{N/2},\hat u_{1,o}^{i - 1} \oplus \hat u_{1,e}^{i - 1}} \right),L_{N/2}^{\left( {\left\lceil {i/2} \right\rceil } \right)}\left( {y_{N/2 + 1}^N,\hat u_{1,e}^{i - 1}} \right),{{\hat u}_{i - 1}}} \right),}&${{\text{when}}\;i\;{\text{is}}\;{\text{even}}}$
\end{numcases}
\label{eq-SC-recursive-formula-LR}
\end{small}
\end{subequations}
, where $f\left( {a,b} \right) = \frac{{a \cdot b + 1}}{{a + b}}\;\;\;\text{and}\;\;\;g\left( {a,b,s} \right) = {a^{1 - 2s}} \cdot b$. We name $L_N^{\left( i \right)}\left( {y_1^N,\hat u_1^{i - 1}} \right)$ as the $i^{\rm{th}}$ LR at length $N$. Its calculation is recursively converted into the calculations of the two LRs at length $N/2$, and the recursion is continued down to the calculations of the $N$ LRs at length 1, i.e.
\begin{eqnarray}
L_1^{\left( 1 \right)}\left( {{y_i}} \right) = \frac{{W\left( {\left. {{y_i}} \right|0} \right)}}{{W\left( {\left. {{y_i}} \right|1} \right)}}, \;\;1 \le i \le N,
\label{eq-channel-LR}
\end{eqnarray}
which can be computed immediately according to the output vector $y_1^N$.

Another two recursive formulas similar to \eqref{eq-SC-recursive-formula-LR} are shown in \eqref{eq-SC-recursive-formula-LLR} and \eqref{eq-SC-recursive-formula-min-sum}, where $F\left( {a,b} \right) = 2{\mathop{\rm arctanh}\nolimits} \left( {\tanh \left( {a/2} \right) \cdot \tanh \left( {b/2} \right)} \right)$, $G\left( {a,b,s} \right) = {\left( { - 1} \right)^s}a + b$, and $\mathbb{F}\left( {a,b} \right) = {\mathop{\rm sgn}} \left( a \right) \cdot {\mathop{\rm sgn}} \left( b \right) \cdot \min \left\{ {\left| a \right|,\left| b \right|} \right\}$. These two recursive formulas are both based on logarithm likelihood ratio (LLR), which is employed frequently by kinds of decoders, because it is always superior to the LR in terms of hardware utilization, computational complexity, and numerical stability\cite{zhang2013low}. \eqref{eq-SC-recursive-formula-min-sum} is also known as min-sum update rule, which further simplifies the implementations of the hyperbolic tangent function and its inverse function in \eqref{eq-SC-recursive-formula-LLR}. Although the recursive formulas \eqref{eq-SC-recursive-formula-LR}, \eqref{eq-SC-recursive-formula-LLR} and \eqref{eq-SC-recursive-formula-min-sum} are different, the dependencies among nodes are the same if we regard a LR or LLR as a node. Without loss of generality, we employ the recursive formula \eqref{eq-SC-recursive-formula-LR} to depict our idea.

\begin{subequations}
\begin{small}
\begin{numcases}{\mathbb{L}_N^{\left( i \right)}\left( {y_1^N,\hat u_1^{i - 1}} \right) = } \label{eq-SC-recursive-formula-LLR-a}
{F\left( {\mathbb{L}_{N/2}^{\left( {\left\lceil {i/2} \right\rceil } \right)}\left( {y_1^{N/2},\hat u_{1,o}^{i - 1} \oplus \hat u_{1,e}^{i - 1}} \right),\mathbb{L}_{N/2}^{\left( {\left\lceil {i/2} \right\rceil } \right)}\left( {y_{N/2 + 1}^N,\hat u_{1,e}^{i - 1}} \right)} \right),} &${{\text{when}}\;i\;{\text{is}}\;{\text{odd}}}\;\;\;\;\;\;\;\;\;$ \\\label{eq-SC-recursive-formula-LLR-b}
{G\left( {\mathbb{L}_{N/2}^{\left( {\left\lceil {i/2} \right\rceil } \right)}\left( {y_1^{N/2},\hat u_{1,o}^{i - 1} \oplus \hat u_{1,e}^{i - 1}} \right),\mathbb{L}_{N/2}^{\left( {\left\lceil {i/2} \right\rceil } \right)}\left( {y_{N/2 + 1}^N,\hat u_{1,e}^{i - 1}} \right),{{\hat u}_{i - 1}}} \right),}&${{\text{when}}\;i\;{\text{is}}\;{\text{even}}}$
\end{numcases}
\label{eq-SC-recursive-formula-LLR}
\end{small}
\end{subequations}

\begin{subequations}
\begin{small}
\begin{numcases}{\mathbb{L}_N^{\left( i \right)}\left( {y_1^N,\hat u_1^{i - 1}} \right) = } \label{eq-SC-recursive-formula-min-sum-a}
{\mathbb{F}\left( {\mathbb{L}_{N/2}^{\left( {\left\lceil {i/2} \right\rceil } \right)}\left( {y_1^{N/2},\hat u_{1,o}^{i - 1} \oplus \hat u_{1,e}^{i - 1}} \right),\mathbb{L}_{N/2}^{\left( {\left\lceil {i/2} \right\rceil } \right)}\left( {y_{N/2 + 1}^N,\hat u_{1,e}^{i - 1}} \right)} \right),} &${{\text{when}}\;i\;{\text{is}}\;{\text{odd}}}\;\;\;\;\;\;\;\;\;$ \\\label{eq-SC-recursive-formula-min-sum-b}
{G\left( {\mathbb{L}_{N/2}^{\left( {\left\lceil {i/2} \right\rceil } \right)}\left( {y_1^{N/2},\hat u_{1,o}^{i - 1} \oplus \hat u_{1,e}^{i - 1}} \right),\mathbb{L}_{N/2}^{\left( {\left\lceil {i/2} \right\rceil } \right)}\left( {y_{N/2 + 1}^N,\hat u_{1,e}^{i - 1}} \right),{{\hat u}_{i - 1}}} \right),}&${{\text{when}}\;i\;{\text{is}}\;{\text{even}}}$
\end{numcases}
\label{eq-SC-recursive-formula-min-sum}
\end{small}
\end{subequations}

According to \eqref{eq-SC-determine} and \eqref{eq-SC-recursive-formula-LR}, in order to estimate ${\hat u_i}$, the computation of the $i^{th}$ LR at length $N$ is firstly activated, which in turn activates the two LRs at length $N/2$. The two LRs at length $N/2$ then activate the four LRs at length $N/4$, which activate the eight LRs at length $N/8$. The process continues till the LRs at certain length, assumed as $N/2^k$, have been estimated. Then the computation is sequentially performed back from length $N/2^{k-1}$ to length $N$, and ${\hat u_i}$ is determined according to \eqref{eq-SC-determine}. In other words, SC decoder achieves the maximized sharing on the calculations of LRs by a recursive way. We name the recursive way as implicitly maximized sharing, because it can not immediately recognize which of LRs could be shared. To achieve an explicitly maximized sharing, the existed implementations firstly calculate the decoding schedule offline by certain algorithm, then store it in the ROM. Different from them, the proposed algorithm obtains the decoding schedule on the fly without introducing any extra delay, which owes to a new-found factor $z_i$.

\section{\label{sec-Derive theoretically from recursive formula}Dependencies among LRs}
To achieve an explicitly maximized sharing on the calculations of LRs, we firstly probe into the recursive formula \eqref{eq-SC-recursive-formula-LR} to explore the dependencies among LRs. It is obvious that there are three operations performing on the decoded bits $\widehat u_1^{i - 1}$ in \eqref{eq-SC-recursive-formula-LR}. They are the XOR between the subvectors with odd indices and even indices, the EXTRACTION of the subvector with even indices, and the EXTRACTION of the last element
, i.e.
\begin{eqnarray}
\begin{array}{*{20}{l}}
{p\left( {\hat u_1^{i - 1}} \right)}& = &{\hat u_{1,o}^{i - 1} \oplus \hat u_{1,e}^{i - 1}}& = &{\hat u_{1,o}^{2\left\lfloor {\frac{{i - 1}}{2}} \right\rfloor } \oplus \hat u_{1,e}^{2\left\lfloor {\frac{{i - 1}}{2}} \right\rfloor }}\\[1mm]
{q\left( {\hat u_1^{i - 1}} \right)}& = &{\hat u_{1,e}^{i - 1}}& = &{\hat u_{1,e}^{2\left\lfloor {\frac{{i - 1}}{2}} \right\rfloor }}\\[1mm]
{r\left( {\hat u_1^{i - 1}} \right)}& = &{\hat u_{i - 1}^{}}.&{}&{}
\end{array}
\label{eq-pqr-operators-definition}
\end{eqnarray}
By these three operations, we deduce which of LRs are used during the calculation of the $i^{\rm{th}}$ LR at length $N$.
\begin{lem} \label{lem-LR-N-i} For any given $1 \le k\le n$, the calculation of the $i^{\rm{th}}$ LR at length $N$, $L_N^{\left( i \right)}\left( {y_1^N,\hat u_1^{i - 1}} \right)$, depends on the calculations of $2^k$ LRs at length ${N \mathord{\left/{\vphantom {N {{2^k}}}} \right.
 \kern-\nulldelimiterspace} {{2^k}}}$,
\begin{eqnarray}
L_{{N \mathord{\left/
 {\vphantom {N {{2^k}}}} \right.
 \kern-\nulldelimiterspace} {{2^k}}}}^{\left( {\left\lceil {{i \mathord{\left/
 {\vphantom {i {{2^k}}}} \right.
 \kern-\nulldelimiterspace} {{2^k}}}} \right\rceil } \right)}\left( {y_{\left( {j - 1} \right) \cdot {N \mathord{\left/
 {\vphantom {N {{2^k}}}} \right.
 \kern-\nulldelimiterspace} {{2^k}}} + 1}^{j \cdot {N \mathord{\left/
 {\vphantom {N {{2^k}}}} \right.
 \kern-\nulldelimiterspace} {{2^k}}}},{h_{j,k}}\left( {\hat u_1^{i - 1}} \right)} \right),\;\;\;1 \leqslant j \leqslant {2^k},
 \label{eq-ith-LR-N-dependent-N2k}
\end{eqnarray}
where ${h_{j,k}}$ is a composite function of functions $p$ and $q$. Specifically, ${h_{j,k}} = {\theta_k} \circ {\theta_{k - 1}} \circ  \cdots  \circ {\theta_2} \circ {\theta_1}$, where
\begin{eqnarray}
{\theta_a} = \left\{ {\begin{array}{*{20}{l}}
  {p,}&{{\text{when}}\;{b_{k - a + 1}} = 0} \\
  {q,}&{{\text{when}}\;{b_{k - a + 1}} = 1}
\end{array}} \right.,\; 1 \le a \le k,
\end{eqnarray}
and ${b_k}{b_{k - 1}} \cdots {b_2}{b_1}$ is the binary expansion of the integer $j-1$.
\end{lem}
\begin{proof}
Please refer to Appendix \ref{appendix-Proof of Lema LRNi}.
\end{proof}

\eqref{eq-ith-LR-N-dependent-N2k} indicates that the LRs at length $N/2^k$ depended by two diverse LRs at length $N$ are different in two items. One is $\left\lceil {i/{2^k}} \right\rceil$, and the other is ${h_{j,k}}\left( {\hat u_1^{i - 1}} \right)$. It is obvious that
\begin{eqnarray}
\left\lceil {i/{2^k}} \right\rceil  \equiv m,\;{\rm{for}}\;\forall \left( {m - 1} \right){2^k} + 1 \le i \le m{2^k}.
\label{eq-i2k-m2k}
\end{eqnarray}
Then how about ${h_{j,k}}\left( {\hat u_1^{i - 1}} \right)$?

\begin{lem} \label{lem-h_j_k} For any given $1 \le k \le n$ and $1 \leqslant j \leqslant {2^k}$,  ${h_{j,k}}\left( {\hat u_1^i} \right)$ is a vector with the length of $\left\lfloor {{i \mathord{\left/ {\vphantom {i {{2^k}}}} \right. \kern-\nulldelimiterspace} {{2^k}}}} \right\rfloor$, denoted as $\left( {{v_1},{v_2}, \cdots ,{v_{\left\lfloor {{i \mathord{\left/ {\vphantom {i {{2^k}}}} \right. \kern-\nulldelimiterspace} {{2^k}}}} \right\rfloor }}} \right)$.  Any element $v_a$ is estimated as follows,
\[{v_a} = \mathop  \oplus \limits_{d \in {D_{j,k,a}}} {{\hat u}_d},\;\;1 \le a \le \left\lfloor {i/{2^k}} \right\rfloor,\]
where
\begin{eqnarray}
{D_{j,k,a}} = \left\{ {d\left| {d = \left( {a - 1} \right) \cdot {2^k} + 1 + {c_k}{c_{k - 1}} \cdots {c_1}} \right.} \right\},
\label{eq-hjk-Djka}
\end{eqnarray}
\begin{eqnarray}
{c_t} = \left\{ {\begin{array}{*{20}{c}}
  {?,}&{when\;{b_{k - t + 1}} = 0} \\
  {1,}&{when\;\;{b_{k - t + 1}} = 1}
\end{array}} \right., 1 \le t \le k
\label{eq-hjk_ci}
\end{eqnarray}
$c_t=?$ indicates that $c_t$ can be 0 and 1, and ${b_k}{b_{k - 1}} \cdots {b_2}{b_1}$ is the binary expansion of the integer $j-1$.
\end{lem}
\begin{proof}
Please refer to Appendix \ref{appendix-Proof of Lema-hjk}.
\end{proof}

Lema \ref{lem-h_j_k} shows that the vector ${h_{j,k}}\left( {\hat u_1^{i-1}} \right)$ is determined by the values of $j$, $k$ and $\left\lfloor {\left(i-1\right)/{2^k}} \right\rfloor$. Since $\left\lfloor {\left( {i - 1} \right)/{2^k}} \right\rfloor  \equiv m - 1$ for all $\left( {m - 1} \right){2^k} + 1 \le i \le m{2^k}$, we have
\begin{eqnarray}
{h_{j,k}}\left( {\hat u_1^{i - 1}} \right) \equiv {h_{j,k}}\left( {\hat u_1^{\left( {m - 1} \right) \cdot {2^k}}} \right),\;{\rm{for}}\;\forall \left( {m - 1} \right){2^k} + 1 \le i \le m{2^k}.
\label{eq-hjk-i-m}
\end{eqnarray}
Combining \eqref{eq-i2k-m2k}, \eqref{eq-hjk-i-m} and Lema \ref{lem-LR-N-i}, it can be concluded that the $2^k$ LRs at length $N$,
\[L_N^{\left( i \right)}\left( {y_1^N,\hat u_1^{i - 1}} \right),\;\left( {m - 1} \right){2^k} + 1 \le i \le m{2^k},\]
share the same $2^k$ LRs at length ${{N \mathord{\left/ {\vphantom {N 2}} \right.  \kern-\nulldelimiterspace} 2}^k}$
\begin{eqnarray}
L_{{N \mathord{\left/
 {\vphantom {N {{2^k}}}} \right.
 \kern-\nulldelimiterspace} {{2^k}}}}^{\left( m \right)}\left( {y_{\left( {j - 1} \right) \cdot {N \mathord{\left/
 {\vphantom {N {{2^k}}}} \right.
 \kern-\nulldelimiterspace} {{2^k}}} + 1}^{j \cdot {N \mathord{\left/
 {\vphantom {N {{2^k}}}} \right.
 \kern-\nulldelimiterspace} {{2^k}}}},{h_{j,k}}\left( {\hat u_1^{\left( {m - 1} \right) \cdot {2^k}}} \right)} \right),\;1 \leqslant j \leqslant {2^k}.
 \label{eq-2k-N2k-sharedby-2kN}
\end{eqnarray}
According to the conclusion, we find a factor defined as Definition \ref{def-Sharing Factor}. It is the key to achieve an explicitly maximized sharing on the calculations of LRs, as shown in the following Theorem \ref{thm-LR-share-maximization}.

\begin{defn} \label{def-Sharing Factor} (\textbf{Sharing Factor}) For any given $1 \le i \le N$, its sharing factor is denoted as $z_i$, which is the number of the consecutive zeros in the end of the binary expansion of the integer $i-1$. It is noted that ${z_i}=n$ when $i=1$.
\end{defn}

In the view of the sharing factor $z_i$, $i-1$ can be rewritten as follows,
\begin{eqnarray}
i - 1 = \left\{ {\begin{array}{*{20}{l}}
  {0,}&{if\;i = 1} \\
  {{m_o} \cdot {2^{{z_i}}},}&{otherwise}
\end{array}} \right.
\label{eq-i-1-zi}
\end{eqnarray}
where $m_o$ is odd. Theorem \ref{thm-LR-share-maximization} details the function of $z_i$ on the explicitly maximized sharing on the calculations of LRs.



\begin{thm} \label{thm-LR-share-maximization} In SC decoder, when ${\hat u_i}$  is estimated,  only the LRs at length $N, {N \mathord{\left/
 {\vphantom {N 2}} \right.
 \kern-\nulldelimiterspace} 2}, \cdots ,{N \mathord{\left/
 {\vphantom {N {{2^{z_i}}}}} \right.
 \kern-\nulldelimiterspace} {{2^{z_i}}}}$ should be calculated, while the LRs at length $N/2^{z_i+1}$, $N/2^{z_i+2}, \cdots, 1$ can be shared, where ${z_i}$ is the sharing factor of $i$. Specifically, the calculations can be performed beginning with the LRs at length ${N \mathord{\left/
 {\vphantom {N {{2^{z_i}}}}} \right.
 \kern-\nulldelimiterspace} {{2^{z_i}}}}$, and in sequence till ending with the LR at length $N$.
\end{thm}
\begin{proof}
Please refer to Appendix \ref{appendix-Proof of Theorem-thm-LR-share-maximization}.
\end{proof}

\section{\label{sec-Proposed Algorithm to Obtain Decoding Schedule}Proposed Decoding Schedule Generating Algorithm}
For the estimation of ${{\hat u}_i}$, all required LRs at length $N/2^k$ are listed in \eqref{eq-2k-N2k-sharedby-2kN}. Hence we just need to determine which of formula $f$ and $g$ is employed to calculate these LRs. If the formula $f$ is used, these calculations are denoted as $f_{N/2^k}$, otherwise denoted as $g_{N/2^k}$. For the sake of brevity, $f_{N/2^k}$ and $g_{N/2^k}$ are both represented by $\gamma_{N/2^k}$, then
\begin{eqnarray}
\gamma_{N/2^k} = \left\{ {\begin{array}{*{20}{l}}
{f_{N/2^k},} & {{\text{when}}\;{ {\left\lceil {i/{2^{k}}} \right\rceil } }\;{\text{is}}\;{\text{odd}}} \\[1mm]
{g_{N/2^k},}& {{\text{when}}\;{ {\left\lceil {i/{2^{k}}} \right\rceil } }\;{\text{is}}\;{\text{even}}}
\end{array}} \right..
\label{eq-gamma-f-g}
\end{eqnarray}
According to Theorem \ref{thm-LR-share-maximization}, the decoding schedule for the estimation of ${{\hat u}_i}$ is $\gamma_{N/2^{z_i}}, \gamma_{N/2^{z_i-1}}, \cdots, \gamma_N$.

By employing the sharing factor $z_i$, we can further simplify the selection of $\gamma_{N/2^k}$ between $f_{N/2^k}$ and $g_{N/2^k}$. According to \eqref{eq-i-1-zi}, we have
\[\left\lceil {{i \mathord{\left/
 {\vphantom {i {{2^k}}}} \right.
 \kern-\nulldelimiterspace} {{2^k}}}} \right\rceil  = \left\{ {\begin{array}{*{20}{l}}
  {\left\lceil {{1 \mathord{\left/
 {\vphantom {1 {{2^k}}}} \right.
 \kern-\nulldelimiterspace} {{2^k}}}} \right\rceil ,}&{if\;i = 1} \\[1mm]
  {\left\lceil {{m_o} \cdot {2^{{z_i} - k}} + {1 \mathord{\left/
 {\vphantom {1 {{2^k}}}} \right.
 \kern-\nulldelimiterspace} {{2^k}}}} \right\rceil ,}&{otherwise}
\end{array}} \right.,\]
where $m_o$ is odd. Obviously, if $i=1$, then $\left\lceil {i/{2^k}} \right\rceil$ is always equal to 1 for all the $k$, otherwise it is even for $k = {z_i}$ and odd for $k < {z_i}$. Here the case of $k > z_i$ are not considered,  because Theorem \ref{thm-LR-share-maximization} shows that the LRs at length $N/2^k$, $k > z_i$, can be shared and need not be calculated. Hence the parity of $\left\lceil {i/{2^k}} \right\rceil$ can be determined as follows,
\begin{eqnarray}
\left\lceil {i/{2^k}} \right\rceil = \left\{ {\begin{array}{*{20}{l}}
  {even,}&{if\; i \ne 1 \; and \; k=z_i} \\
  {odd,}&{otherwise}
\end{array}} \right.,
\label{eq-Discriminant-formula-for-parity-m-3}
\end{eqnarray}
and the selection of $\gamma_{N/2^k}$ can be rewritten as follows,
\begin{eqnarray}
\gamma_{N/2^k} = \left\{ {\begin{array}{*{20}{l}}
  {g_{N/2^k},}&{if\; i \ne 1 \; and \; k=z_i} \\
  {f_{N/2^k},}&{otherwise}
\end{array}} \right. .
\end{eqnarray}
Another method to determine the selection of $\gamma_{N/2^k}$ was also presented in \cite{Leroux2013A,Le2015Multi}, i.e.
\begin{eqnarray}
\gamma_{N/2^k} = \left\{ {\begin{array}{*{20}{l}}
  {g_{N/2^k},}&{when\;\left( {i - 1} \right)\& {2^k} =  = 1} \\
  {f_{N/2^k},}&{when\;\left( {i - 1} \right)\& {2^k} =  = 0}
\end{array}} \right..
\label{eq-Discriminant-formula-for-parity-m-1}
\end{eqnarray}
In their method, for each $1 \le i \le N$, the selections should be performed for all $k$. While our method shows that, for $2 \le i \le N$, the formula $g$ is just employed one time, i.e. $k=z_i$, and for $i=1$, the formula $g$ does not be employed.

In a word, the proposed decoding schedule algorithm is summarized as follows. The SC decoder successively estimates the transmitted bits $\hat u_1^N$: for the estimation of $\hat u_1$, $f_1, f_2, f_4, \cdots, f_N$ are performed in sequence, and for the estimation of $\hat u_i$ ($i > 1$), $g_{N/2^{z_i}}, f_{N/2^{z_i-1}}, \cdots, f_N$ are performed in sequence.  An example of $N=8$ is shown in Table \ref{tbl-decoding-schedule-of-SC-decoder-for-polar-codes-with-N8}. The first line is clock cycle, the second line is the entries of the decoding schedule, and the third line is the output of ${\hat u}_i$.

\begin{table}[htbp]
\centering
\caption{\label{tbl-decoding-schedule-of-SC-decoder-for-polar-codes-with-N8}The decoding schedule of SC decoder for polar codes with $N=8$.}
\begin{tabular}{c|c|c|c|c|c|c|c|c|c|c|c|c|c|c|c}
\hline
        CC &          1 &          2 &          3 &          4 &          5 &          6 &          7 &          8 &          9 &         10 &         11 &         12 &         13 &         14 &         15 \\
\hline
        Entry &         $f_1$ &         $f_2$ &         $f_4$ &         $f_8$ &         $g_8$ &         $g_4$ &         $f_8$ &         $g_8$ &         $g_2$ &         $f_4$ &         $f_8$ &         $g_8$ &         $g_4$ &         $f_8$ &         $g_8$ \\
\hline
        ${\hat u}_i$ &        N/A &        N/A &        N/A &         ${\hat u}_1$ &         ${\hat u}_2$ &        N/A &         ${\hat u}_3$ &         ${\hat u}_4$ &        N/A &        N/A &         ${\hat u}_5$ &         ${\hat u}_6$ &        N/A &         ${\hat u}_7$ &         ${\hat u}_8$ \\
\hline
\end{tabular}
\end{table}

In order to generate the decoding schedule online, we would like to calculate $z_i$ on the fly. According to Bit Twiddling Hacks\footnote{http://graphics.stanford.edu/$\sim$seander/bithacks.html.}, the online calculation of $z_i$ is feasible. An illustration with a multiply and a lookup table is shown in Algorithm \ref{algo-An illustration of calculating zi}. The algorithm works for any input $2 \le i \le 2^{32}$. For the case $i=1$, $z_1$ is set to $n$. The codelength $N=2^{32}$ is enough for almost all practical polar codes, whose codelengths usually are about $2^{20}$ bits. Since the calculation of $z_i$ is so simple, its delay can be easily eliminated by packing it into the estimations of ${\hat u_{j}}$, ($j < i$). Hence $z_i$ can be calculated on the fly without introducing any extra delay. The proposed algorithm thereby generates decoding schedule online without introducing any extra delay.

\begin{algorithm}[h]
\caption{An illustration of calculating $z_i$}
\label{algo-An illustration of calculating zi}
\begin{algorithmic}[1]
\REQUIRE ~~~~ $i$;
\ENSURE ~~ $z_i$;

    \STATE static const int LT[32] = \{ 0, 1, 28, 2, 29, 14, 24, 3, 30, 22, 20, 15, 25, 17, 4, 8, 31, 27, 13, 23, 21, 19, 16, 7, 26, 12, 18, 6, 11, 5, 10, 9 \};\\
    \STATE int $v=i-1$;\\
    \STATE $z_i$ = LT[((uint32\_t)(($v$ \& -$v$) $*$ 0x077CB531)) $>>$ 27];\\

\end{algorithmic}
\end{algorithm}

\section{\label{sec-Comparisons}Comparisons}
To the best of our knowledge, there are two algorithms generating the decoding schedule of SC decoder, both of which calculate the decoding schedule offline and store it in the ROM. One is based on the depth-first traversal of the scheduling tree\cite{Alamdar2011A,huang2012latency,Sarkis2013Increasing,YouZhe2014An,Le2015Multi}, as shown in Figure \ref{Fig-Scheduling-tree}. The other is based on the recursive construction\cite{Chuan2012Reduced}, as shown in Algorithm \ref{algo-Recursive construction of decoding schedule}. Obviously, the proposed algorithm based on the sharing factor $z_i$ is more simple.

\begin{figure}[htbp]

\begin{minipage}[b]{0.45\linewidth}
  \centerline{\includegraphics[width=1.\textwidth]{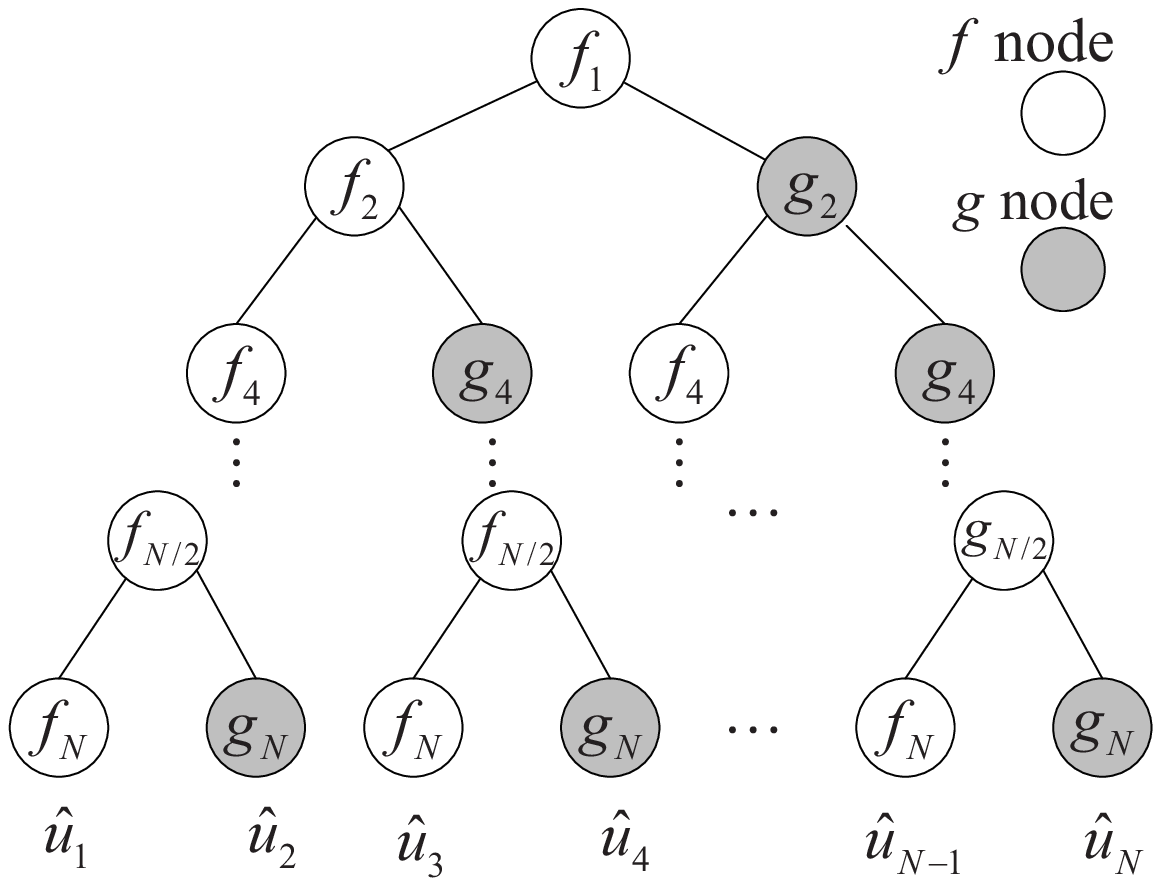}}
  \centerline{(a)}
\end{minipage}
\hfill
\begin{minipage}[b]{.45\linewidth}
  \centerline{\includegraphics[width=1.\textwidth]{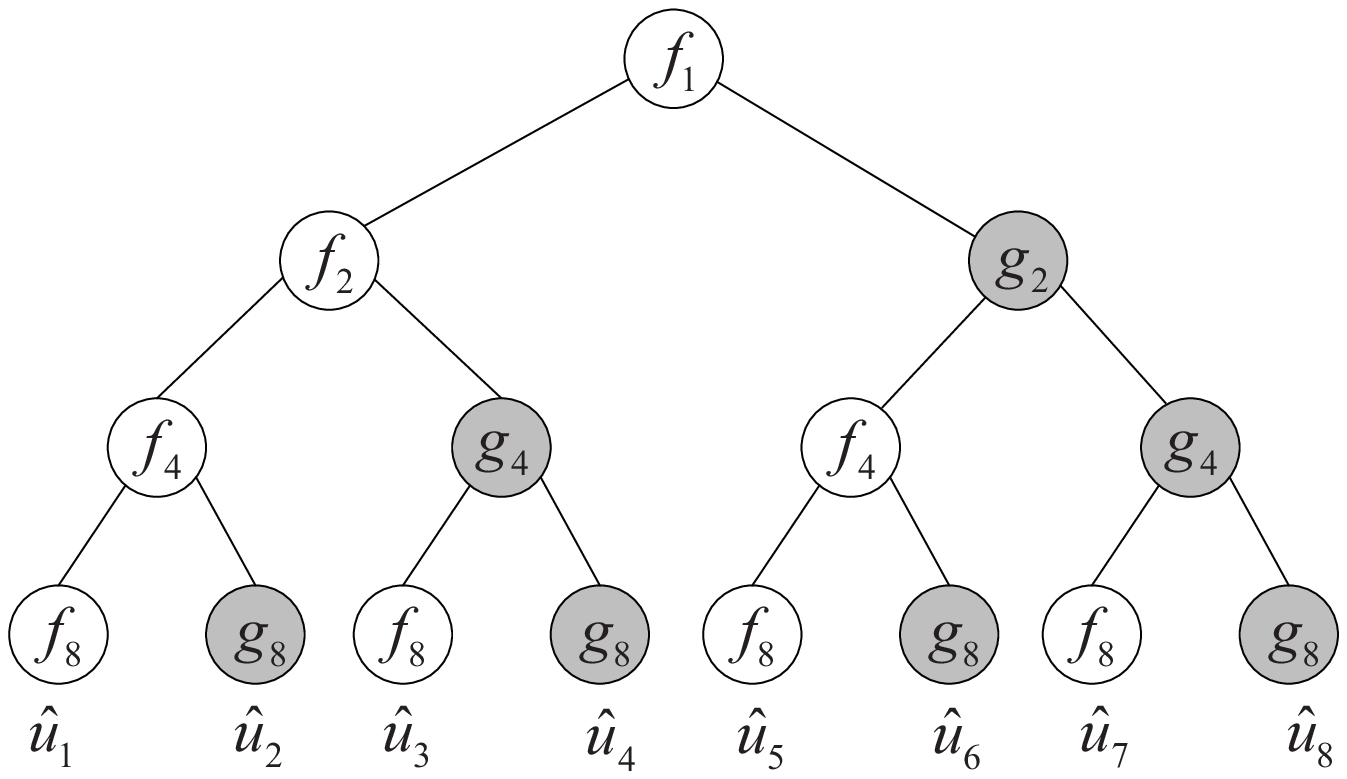}}
  \centerline{(b)}
\end{minipage}
\caption{\label{Fig-Scheduling-tree} Scheduling tree of the SC decoder for polar codes, whose depth-first traversal generates the decoding scheduling. (a) General scheduling tree. (b) Example of $N=8$.}
\end{figure}

\begin{algorithm}[htbp]
\caption{Recursive-construction based decoding schedule generating algorithm\cite{Chuan2012Reduced}}
\label{algo-Recursive construction of decoding schedule}
\begin{algorithmic}[1]
\REQUIRE ~~~~~  Codelength $N$;
\ENSURE ~~ Decoding schedule $DS$;

    \STATE $DS$ = $NULL$;\\
    \FOR {$i=n$, $i--$, 1}
        \STATE $DS1$ = [$f_{2^i}$, $DS$];\\
        \STATE $DS2$ = [$g_{2^i}$, $DS$];\\
        \STATE $DS$ = [$DS1$, $DS2$];\\
    \ENDFOR
    \STATE $DS = [f_1, DS]$;
    \STATE Output $DS$;\\

\end{algorithmic}
\end{algorithm}

We compare the three algorithms in terms of online (No/Yes), extra delay(No/Yes), and memory. The indicator of online measures whether the decoding schedule is generated on the fly or not. The indicator of extra delay measures whether the generation of decoding schedule introduces extra delay into the decoding process. The indicator of memory shows the number of the storage taken by the decoding schedule. The comparison results are listed in Table \ref{tbl-Comparison-for-different-decoding-schedule algorithms}.

\begin{table}[htbp]
\centering
\caption{\label{tbl-Comparison-for-different-decoding-schedule algorithms}Comparisons among different decoding schedule generating algorithms.}
\begin{tabular}{l|ccc}
\hline
           & Online & Extra Delay &     Memory (bit) \\
\hline
Scheduling Tree\cite{Alamdar2011A,huang2012latency,Sarkis2013Increasing,YouZhe2014An,Le2015Multi} &         No &         No &  $\left( {2N - 1} \right) {{{\log }_2}\left( {2n + 1} \right)} $          \\
\hline
Recursive Construction\cite{Chuan2012Reduced} &         No &         No &    $\left( {2N - 1} \right) {{{\log }_2}\left( {2n + 1} \right)} $        \\
\hline
  Proposed &        \textbf{Yes} &         No &      \textbf{160}      \\
\hline
\end{tabular}

\end{table}

Since the algorithms based on the scheduling tree and recursive construction both generate the decoding schedule offline and store it in the ROM, they are not online and do not introduce any extra delay. For the two existed algorithms, the required memory is equal to the product of the number of the entries of decoding schedule and the number of bits to represent each entry. It is obvious that the number of the entries of decoding schedule is the total nodes of schedule tree, i.e. $\sum\nolimits_{k = 0}^n {{2^k}}  = 2N - 1$. Each entry of decoding schedule can be represented at least by $ {{{\log }_2}\left( {2n + 1} \right)}$, because there are $2n+1$ different entries, i.e. $f_1,\cdots,f_{N/2}, f_{N}, g_2,\cdots, g_{N/2}, g_{N}$. Hence the memory taken by the two existed algorithms are both $\left( {2N - 1} \right) {{{\log }_2}\left( {2n + 1} \right)} $. If Algorithm \ref{algo-An illustration of calculating zi} is employed to calculate $z_i$, as mentioned above, the proposed algorithm can generate the decoding schedule on the fly without introducing any extra delay. Since only a lookup table needs to be stored during the running of the proposed algorithm, its required memory is constant, i.e. 160bits, which is far less than the memory taken by the two existed algorithms, especially for a large $N$. Usually, in order to achieve the channel capacity, the codelength of polar codes should be at least $2^{20}$ bits\cite{Leroux2013A,Yoo2015Partially}, i.e. $N \ge 2^{20}$.

\section{\label{sec-Conclusion}Conclusions}
Thanks to the new-found factor $z_i$, we have proposed a new decoding schedule generating algorithm, which is superior to the existed algorithms in two aspects. The first is that the existed algorithms are too complex, which are either based on the depth-first traversal of the scheduling tree or based on the recursive construction. While the proposed algorithm skillfully computes decoding schedule by the sharing factor $z_i$, which can be calculated easily with Bit Twiddling Hacks. The second is that the existed algorithms obtain decoding schedule offline and consume at least $\left( {2N - 1} \right) {{{\log }_2}\left( {2n + 1} \right)} $ bits to store it. While the proposed algorithm generates decoding schedule on the fly, and just requires 160 bits during the generation of the decoding schedule. These advantages are helpful to decrease the implementation complexity of SC decoder.

%

\section*{Acknowledgment}
The authors gratefully acknowledge Wu Xianyan, Sang Jianzhi and Mao Haokun from Harbin Institute of Technology for many useful discussions on this paper.
This work is supported by the National Natural Science Foundation of China (Grant Number:
61471141, 61361166006, 61301099) and the Fundamental Research Funds for the Central Universities (Grant Number:
HIT. KISTP. 201416, HIT. KISTP. 201414).

 \begin{appendices}
\section{Proof of Lema \ref{lem-LR-N-i}}
\label{appendix-Proof of Lema LRNi}
\begin{proof}

\textbf{\emph{Basis Step:}} We start with the case $k=1$. According to (\ref{eq-SC-recursive-formula-LR}), the calculation of the $i^{th}$ LR at length $N$, $L_N^{\left( i \right)}\left( {y_1^N,\hat u_1^{i - 1}} \right)$, depends on the calculations of two LRs at length ${N \mathord{\left/ {\vphantom {N 2}} \right. \kern-\nulldelimiterspace} 2}$ as follows,

\[L_{{N \mathord{\left/
 {\vphantom {N 2}} \right.
 \kern-\nulldelimiterspace} 2}}^{\left( {\left\lceil {{i \mathord{\left/
 {\vphantom {i 2}} \right.
 \kern-\nulldelimiterspace} 2}} \right\rceil } \right)}\left( {y_1^{{N \mathord{\left/
 {\vphantom {N 2}} \right.
 \kern-\nulldelimiterspace} 2}},p\left( {\hat u_1^{i - 1}} \right)} \right),\;\;L_{{N \mathord{\left/
 {\vphantom {N 2}} \right.
 \kern-\nulldelimiterspace} 2}}^{\left( {\left\lceil {{i \mathord{\left/
 {\vphantom {i 2}} \right.
 \kern-\nulldelimiterspace} 2}} \right\rceil } \right)}\left( {y_{{N \mathord{\left/
 {\vphantom {N 2}} \right.
 \kern-\nulldelimiterspace} 2} + 1}^N,q\left( {\hat u_1^{i - 1}} \right)} \right).\]
Let ${h_{1,1}} = p$ and ${h_{2,1}} = q$, then the lema for the case $k=1$ is true.

\textbf{\emph{Inductive Step:}} Now we assume the truth of the case $k=m$. That is that the calculation of the $i^{th}$ LR at length $N$ depends on the calculations of $2^m$ LRs at length ${N \mathord{\left/ {\vphantom {N 2}} \right. \kern-\nulldelimiterspace} {2^m}}$ as follows,
\begin{eqnarray}
L_{{N \mathord{\left/
 {\vphantom {N {{2^m}}}} \right.
 \kern-\nulldelimiterspace} {{2^m}}}}^{\left( {\left\lceil {{i \mathord{\left/
 {\vphantom {i {{2^m}}}} \right.
 \kern-\nulldelimiterspace} {{2^m}}}} \right\rceil } \right)}\left( {y_{\left( {j - 1} \right) \cdot {N \mathord{\left/
 {\vphantom {N {{2^m}}}} \right.
 \kern-\nulldelimiterspace} {{2^m}}} + 1}^{j \cdot {N \mathord{\left/
 {\vphantom {N {{2^m}}}} \right.
 \kern-\nulldelimiterspace} {{2^m}}}},{h_{j,m}}\left( {\hat u_1^{i - 1}} \right)} \right),\;\;\;1 \leqslant j \leqslant {2^m},
 \label{eq-L-N-2-m-LRs}
\end{eqnarray}
where
\begin{eqnarray}
{h_{j,m}} = {\theta_m} \circ {\theta_{m - 1}} \circ  \cdots  \circ {\theta_2} \circ {\theta_1}
\label{eq-h-j-m}
\end{eqnarray}
and
\begin{eqnarray}
{\theta_a} = \left\{ {\begin{array}{*{20}{l}}
  {p,}&{{\text{when}}\;{b_{m - a + 1}} = 0} \\
  {q,}&{{\text{when}}\;{b_{m - a + 1}} = 1}
\end{array}} \right., \; 1 \le a \le m.
\label{eq-f-i-m}
\end{eqnarray}
Here ${b_m}{b_{m - 1}} \cdots {b_2}{b_1}$ is the binary expansion of the integer $j-1$.
According to \eqref{eq-SC-recursive-formula-LR}, if $\left\lceil {{i \mathord{\left/ {\vphantom {i {{2^m}}}} \right. \kern-\nulldelimiterspace} {{2^m}}}} \right\rceil $ is odd, then each item in (\ref{eq-L-N-2-m-LRs}) is calculated as (\ref{eq-L-N-m+1-odd}), otherwise it is calculated as (\ref{eq-L-N-m+1-even}).

\begin{subequations}
\begin{small}
\begin{numcases}{{L_{N/{2^m}}^{\left( \left\lceil {i/{2^m}} \right\rceil \right)}\left( {y_{\left( {j - 1} \right) \cdot N/{2^m} + 1}^{j \cdot N/{2^m}},{h_{j,m}}\left( {\hat u_1^{i - 1}} \right)} \right)} = } \label{eq-L-N-m+1-odd}
{f\left( {\alpha, \beta} \right),} &${{\text{when}}\;{ {\left\lceil {i/{2^{m}}} \right\rceil } }\;{\text{is}}\;{\text{odd}}}\;\;\;\;\;\;\;\;\;\;\;\;\;$ \\\label{eq-L-N-m+1-even}
{g\left( {\alpha, \beta, r\left( {{h_{j,m}}\left( {\hat u_1^{i - 1}} \right)} \right)} \right),}&${{\text{when}}\;{ {\left\lceil {i/{2^{m}}} \right\rceil } }\;{\text{is}}\;{\text{even}}}$
\end{numcases}
\label{eq-L-N-m+1}
\end{small}
\end{subequations}
where
\[\begin{small}\begin{array}{l}
\alpha  = L_{N/{2^{m + 1}}}^{\left( {\left\lceil {i/{2^{m + 1}}} \right\rceil } \right)}\left( {y_{\left( {2j - 2} \right) \cdot N/{2^{m + 1}} + 1}^{\left( {2j - 1} \right) \cdot N/{2^{m + 1}}},p\left( {{h_{j,m}}\left( {\hat u_1^{i - 1}} \right)} \right)} \right), \;\;
\beta  = L_{N/{2^{m + 1}}}^{\left( {\left\lceil {i/{2^{m + 1}}} \right\rceil } \right)}\left( {y_{\left( {2j - 1} \right) \cdot N/{2^{m + 1}} + 1}^{ {2j}  \cdot N/{2^{m + 1}}},q\left( {{h_{j,m}}\left( {\hat u_1^{i - 1}} \right)} \right)} \right).
\end{array}\end{small}\]

According to \eqref{eq-L-N-m+1}, It is obvious that the calculation of a LR at length ${N \mathord{\left/ {\vphantom {N {{2^m}}}} \right. \kern-\nulldelimiterspace} {{2^m}}}$,
\[L_{{{{N \mathord{\left/ {\vphantom {N 2}} \right. \kern-\nulldelimiterspace} 2}}^m}}^{\left( {\left\lceil {{i \mathord{\left/ {\vphantom {i {{2^m}}}} \right. \kern-\nulldelimiterspace} {{2^m}}}} \right\rceil } \right)}\left( {y_{\left( {j - 1} \right) \cdot {{{N \mathord{\left/ {\vphantom {N 2}} \right. \kern-\nulldelimiterspace} 2}}^m} + 1}^{j \cdot {{{N \mathord{\left/ {\vphantom {N 2}} \right. \kern-\nulldelimiterspace} 2}}^m}},{h_{j,m}}\left( {\hat u_1^{i - 1}} \right)} \right),\]
depends on the calculations of two LRs at length ${N \mathord{\left/ {\vphantom {N {{2^{m + 1}}}}} \right. \kern-\nulldelimiterspace} {{2^{m + 1}}}}$ as follows,

\begin{small}
\[L_{{N \mathord{\left/
 {\vphantom {N {{2^{m + 1}}}}} \right.
 \kern-\nulldelimiterspace} {{2^{m + 1}}}}}^{\left( {\left\lceil {{i \mathord{\left/
 {\vphantom {i {{2^{m + 1}}}}} \right.
 \kern-\nulldelimiterspace} {{2^{m + 1}}}}} \right\rceil } \right)}\left( {y_{\left( {l - 1} \right) \cdot {N \mathord{\left/
 {\vphantom {N {{2^{m + 1}}}}} \right.
 \kern-\nulldelimiterspace} {{2^{m + 1}}}} + 1}^{l \cdot {N \mathord{\left/
 {\vphantom {N {{2^{m + 1}}}}} \right.
 \kern-\nulldelimiterspace} {{2^{m + 1}}}}},{h_{l,m + 1}}\left( {\hat u_1^{i - 1}} \right)} \right),\;l \in \left\{ {2j - 1,2j} \right\}\]
 \end{small}
, where
\begin{eqnarray}
{h_{l,m + 1}}\left( {\hat u_1^{i - 1}} \right) = \left\{ {\begin{array}{*{20}{l}}
  {p\left( {{h_{j,m}}\left( {\hat u_1^{i - 1}} \right)} \right),}&{{\text{when}}\;l = 2j - 1} \\
  {q\left( {{h_{j,m}}\left( {\hat u_1^{i - 1}} \right)} \right),}&{{\text{when}}\;l = 2j}
\end{array}} \right..
\label{eq-h-l-m+1}
\end{eqnarray}
That is
\[{h_{l,m + 1}} = \left\{ {\begin{array}{*{20}{l}}
  {p \circ {f_m} \circ  \cdots  \circ {f_2} \circ {f_1},}&{when\;l = 2j - 1} \\
  {q \circ {f_m} \circ  \cdots  \circ {f_2} \circ {f_1},}&{when\;l = 2j}
\end{array}} \right.\]

Since $j \in \left\{ {1,2, \cdots ,{2^m}} \right\}$, it can be inferred that $l \in \left\{ {1,2, \cdots ,{2^{m + 1}}} \right\}$ and
\begin{eqnarray}
l - 1 = \left\{ {\begin{array}{*{20}{l}}
  {{b_m}{b_{m - 1}} \cdots {b_1}0,}&{when\;l = 2j - 1} \\
  {{b_m}{b_{m - 1}} \cdots {b_1}1,}&{when\;l = 2j}
\end{array}} \right..
\label{eq-l-1-binary-expanse}
\end{eqnarray}
Hence the lema for the case $k=m+1$ is true.

Consequently, by the Principle of Finite Induction, the lema is proved.

\end{proof}

\section{Proof of Lema \ref{lem-h_j_k}}
\label{appendix-Proof of Lema-hjk}
\begin{proof}

\textbf{\emph{Basis Step:}} We start with the case $k=1$. In this case $j \in \left\{ {1,2} \right\}$, so we only need to consider ${h_{1,1}}\left( {\hat u_1^i} \right)$ and ${h_{2,1}}\left( {\hat u_1^i} \right)$. Since
\begin{eqnarray}
\begin{gathered}
  {h_{1,1}}\left( {\hat u_1^i} \right) = p\left( {\hat u_1^i} \right) = \hat u_{1,o}^{2\left\lfloor {{i \mathord{\left/
 {\vphantom {i 2}} \right.
 \kern-\nulldelimiterspace} 2}} \right\rfloor } \oplus \hat u_{1,e}^{2\left\lfloor {{i \mathord{\left/
 {\vphantom {i 2}} \right.
 \kern-\nulldelimiterspace} 2}} \right\rfloor } \hfill \\
  {h_{2,1}}\left( {\hat u_1^i} \right) = q\left( {\hat u_1^i} \right) = \hat u_{1,e}^{2\left\lfloor {{i \mathord{\left/
 {\vphantom {i 2}} \right.
 \kern-\nulldelimiterspace} 2}} \right\rfloor } \hfill \\
\end{gathered},
\label{eq-h11-h21}
\end{eqnarray}
it is obvious that their lengths are both $\left\lfloor {{i \mathord{\left/ {\vphantom {i 2}} \right. \kern-\nulldelimiterspace} 2}} \right\rfloor $. For any given $1 \le a \le \left\lfloor {i/2} \right\rfloor $, we have
\[\begin{gathered}
  \begin{array}{*{20}{l}}
  {{D_{1,1,a}}}& = &{\left\{ {2a - 1,2a} \right\}} \\
  {}& = &{\left\{ {d\left| {d = \left( {a - 1} \right) \cdot 2 + 1 + ?} \right.} \right\}}
\end{array} \hfill \\
  \begin{array}{*{20}{l}}
  {{D_{2,1,a}}}& = &{\left\{ {2a} \right\}} \\
  {}& = &{\left\{ {d\left| {d = \left( {a - 1} \right) \cdot 2 + 1 + 1} \right.} \right\}}.
\end{array} \hfill \\
\end{gathered}\]
Hence the lema for the case $k=1$ is true.

\textbf{\emph{Inductive Step:}} Now we assume the truth of the case $k=m$. Then we have
\[{h_{j,m}}\left( {\hat u_1^i} \right) = \left( {{v_1},{v_2}, \cdots ,{v_{{n_1}}}} \right),\;{n_1} = \left\lfloor {{i \mathord{\left/ {\vphantom {i {{2^m}}}} \right. \kern-\nulldelimiterspace} {{2^m}}}} \right\rfloor, \]
where
${v_a} = \mathop  \oplus \limits_{d \in {D_{j,m,a}}} {{\hat u}_d},$
\[{D_{j,m,a}} = \left\{ {d\left| {d = \left( {a - 1} \right) \cdot {2^m} + 1 + {c_m}{c_{m - 1}} \cdots {c_1}} \right.} \right\},\]
\[{c_t} = \left\{ {\begin{array}{*{20}{c}}
  {?,}&{when\;{b_{m - t + 1}} = 0} \\
  {1,}&{when\;\;{b_{m - t + 1}} = 1}
\end{array}} \right., 1 \le t \le m,\]
and ${b_m}{b_{m - 1}} \cdots {b_2}{b_1}$ is the binary expansion of the integer $j-1$.

\textbf{(a)} For any given $1 \leqslant l \leqslant {2^{m + 1}}$, when it is odd, according to (\ref{eq-h-l-m+1}) we have
\[{h_{l,m + 1}}\left( {\hat u_1^i} \right) = p\left( {v_1^{{n_1}}} \right) = w_1^{{n_2}}.\]
Then ${n_2} = \left\lfloor {{{{n_1}} \mathord{\left/  {\vphantom {{{n_1}} 2}} \right. \kern-\nulldelimiterspace} 2}} \right\rfloor  = \left\lfloor {{i \mathord{\left/ {\vphantom {i {{2^{m + 1}}}}} \right. \kern-\nulldelimiterspace} {{2^{m + 1}}}}} \right\rfloor $. For any element $w_a$, $1\le a \le n_2$, we have
\[\begin{array}{*{20}{l}}
  {{w_a}} = {{v_{2a - 1}} \oplus {v_{2a}}} = {\mathop  \oplus \limits_{d \in {D_{j,m,2a - 1}} \cup {D_{j,m,2a}}} {{\hat u}_d}}  \buildrel \Delta \over = \mathop  \oplus \limits_{d \in {D_{l,m + 1,a}}} {\hat u_d}
\end{array}\]
and
\[\begin{small}\begin{array}{*{20}{l}}
{} & {D_{l,m + 1,a}}\\
{} =& {D_{j,m,2a - 1}} \cup {D_{j,m,2a}}\\
  {} =& \left\{ {d\left| {d = \left( {2a - 2} \right) \cdot {2^m} + 1 + {c_m}{c_{m - 1}} \cdots {c_1}} \right.} \right\} \cup \left\{ {d\left| {d = \left( {2a - 1} \right) \cdot {2^m} + 1 + {c_m}{c_{m - 1}} \cdots {c_1}} \right.} \right\}  \\
 {} =& \left\{ {d\left| {d = \left( {a - 1} \right) \cdot {2^{m + 1}} + 1 + 0{c_m}{c_{m - 1}} \cdots {c_1}} \right.} \right\} \cup \left\{ {d\left| {d = \left( {a - 1} \right) \cdot {2^{m + 1}} + 1 + 1{c_m}{c_{m - 1}} \cdots {c_1}} \right.} \right\} \\
 {} =& \left\{ {d\left| {d = \left( {a - 1} \right) \cdot {2^{m + 1}} + 1 + ?{c_m}{c_{m - 1}} \cdots {c_1}} \right.} \right\}.
\end{array}\end{small}\]
Hence the lema for the case that $k=m+1$ and $l$ is odd is true.

\textbf{(b)} In a similar way, the case that $k=m+1$ and $l$ is even can also be proved.

Hence, combining \textbf{(a)} and \textbf{(b)}, the lema is inferred to be true for the case $k=m+1$.

Consequently, by the Principle of Finite Induction, the lema is proved.

\end{proof}

\section{Proof of Theorem \ref{thm-LR-share-maximization}}
\label{appendix-Proof of Theorem-thm-LR-share-maximization}
\begin{proof} According to \eqref{eq-SC-determine}, to estimate ${\hat u_i}$ is to calculate the $i^{th}$ LRs at length $N$, $L_N^{\left( i \right)}\left( {y_1^N,\hat u_1^{i - 1}} \right)$. We firstly prove the proposition that the $i^{th}$ and ${\left( {i - 1} \right)^{th}}$ LR at length $N$ can not share the same $2^k$ LRs at length ${{N \mathord{\left/ {\vphantom {N 2}} \right. \kern-\nulldelimiterspace} 2}^k}$ if and only if there exists an integer $m$ satisfying that $i-1=m \cdot {2^k}$. (i) If there exists an integer $m$ satisfying that $i-1=m \cdot {2^k}$, then
\[\begin{array}{l}
i - 1 \in \left\{ {\left( {m - 1} \right) \cdot {2^k} + 1, \cdots ,m \cdot {2^k}} \right\}\\
i \in \left\{ {m \cdot {2^k} + 1, \cdots ,\left( {m + 1} \right) \cdot {2^k}} \right\}
\end{array},\]
which means the ${\left( {i - 1} \right)^{th}}$ and $i^{th}$ LR at length $N$ depend on two different groups of the $2^k$ LRs at length $N/2^k$. (ii) If there does not exists an integer $m$ satisfying that $i-1=m \cdot {2^k}$, then $i-1$ is expressed as $i - 1 = a \cdot {2^k} + b$ where $a$ and $b$ are both integers and $1\le b \le 2^k-1$. So
\[i - 1,i \in \left\{ {a \cdot {2^k} + 1, \cdots ,\left( {a + 1} \right) \cdot {2^k}} \right\},\]
which means the ${\left( {i - 1} \right)^{th}}$ and $i^{th}$ LR at length $N$ depend on the same $2^k$ LRs at length $N/2^k$.

Now we employ the newly proved proposition to show the truth of the theorem.

\textbf{(a)} If $i=1$, then for all $0 \le k \le n=z_1$ there exists the integer 0 satisfying that $i-1=0 \cdot {2^{k}}$, which means in this case all the LRs at all the lengths should be estimated. Since the $2^{z_1}$ LRs at length 1 are channel LRs, \eqref{eq-channel-LR} indicates that they can be estimated immediately. Afterwards the LRs at length $2, 4, \cdots , N$ can be calculated in sequence according to (\ref{eq-SC-recursive-formula-LR}).

\textbf{(b)} Otherwise, for any given integer $k > {z_i}$, it does not exist any integer $m$ satisfying that $i-1=m \cdot {2^{k}}$, which means that the $i^{th}$ and ${\left( {i - 1} \right)^{th}}$ LR at length $N$ share the same $2^{k}$ LRs at length ${{N \mathord{\left/ {\vphantom {N 2}} \right. \kern-\nulldelimiterspace} 2}^{k}}$ when $k > {z_i}$. On the other hand, for any given integer $k \le {z_i}$, there exists the integer $m_i={2^{{z_i} - k}}{m_o}$ satisfying that $i-1=m_i \cdot {2^{k}}$, which means that the $i^{th}$ and ${\left( {i - 1} \right)^{th}}$ LR at length $N$ can not share the same $2^{k}$ LRs at length ${{N \mathord{\left/ {\vphantom {N 2}} \right. \kern-\nulldelimiterspace} 2}^{k}}$ when $k \le {z_i}$. Therefore only the LRs at length $N, {N \mathord{\left/
 {\vphantom {N 2}} \right. \kern-\nulldelimiterspace} 2}, \cdots ,{N \mathord{\left/ {\vphantom {N {{2^{z_i}}}}} \right. \kern-\nulldelimiterspace} {{2^{z_i}}}}$ should be calculated. Since the shared $2^{{z_i}+1}$ LRs at length ${{N \mathord{\left/ {\vphantom {N 2}} \right. \kern-\nulldelimiterspace} 2}^{{z_i} + 1}}$ have been calculated during the estimation of ${{\hat u}_{i - 1}}$ \footnote{More precisely, the shared $2^{{z_i}+1}$ LRs at length ${{N \mathord{\left/ {\vphantom {N 2}} \right. \kern-\nulldelimiterspace} 2}^{{z_i} + 1}}$ are either calculated during the estimation of ${{\hat u}_{i - 1}}$, or also shared by ${{\hat u}_{i - 1}}$ and ${{\hat u}_{i - 2}}$ . In any case, the shared $2^{{z_i}+1}$ LRs have already been calculated.}, the $2^{z_i}$ LRs at length ${{N \mathord{\left/ {\vphantom {N 2}} \right. \kern-\nulldelimiterspace} 2}^{z_i}}$ can be directly computed according to (\ref{eq-SC-recursive-formula-LR}). Afterwards the LRs at length ${{N \mathord{\left/
 {\vphantom {N 2}} \right. \kern-\nulldelimiterspace} 2}^{{z_i} - 1}},{{N \mathord{\left/ {\vphantom {N 2}} \right.
 \kern-\nulldelimiterspace} 2}^{{z_i} - 2}} \cdots , N$ can be calculated in sequence according to (\ref{eq-SC-recursive-formula-LR}).

Combining \textbf{(a)} and \textbf{(b)}, the theorem is proved.
\end{proof}

\end{appendices}

\bibliographystyle{IEEEtran}
\bibliography{IEEEabrv,ref}

\end{document}